\documentclass[acmsmall, screen, nonacm]{acmart}
\newtoggle{TR}
\toggletrue{TR}

\makeatletter                   
\def\mdseries@tt{m}             
\makeatother                    
\usepackage[plain]{fancyref}
\usepackage[draft=true]{minted} 
\usepackage{color}
\usepackage{hyperref}           
\hypersetup{
    colorlinks=true,
    linkcolor=blue,
    filecolor=red,      
    urlcolor=magenta,
    breaklinks=true,            
}
\usepackage{breakurl}           

\copyrightyear{2019} 
\acmYear{2019} 
\acmConference[TyDe '19]{Proceedings of the 4th ACM SIGPLAN International Workshop on Type-Driven Development}{August 18, 2019}{Berlin, Germany}
\acmBooktitle{Proceedings of the 4th ACM SIGPLAN International Workshop on Type-Driven Development (TyDe '19), August 18, 2019, Berlin, Germany}
\acmPrice{15.00}
\acmDOI{10.1145/3331554.3342608}
\acmISBN{978-1-4503-6815-5/19/08}
\setcopyright{acmlicensed}

\usepackage{pmpaper}
\usepackage{forest}
\usepackage[utf8]{inputenc}
\usepackage[T1]{fontenc}

\newminted{haskell}{}
\newmintinline[hask]{haskell}{}

\DeclareMathOperator{\freshv}{fresh}
\newcommand{\ol}[1]{\overline{#1}}
\newcommand{\many}[2]{\overline{#2}^{#1}}
\newcommand{\kwd}[1]{\ensuremath{\mathbf{#1}\xspace}}

\newcommand{\scythe}{\textsc{Scythe}\xspace}

\begin{document}
\iftoggle{TR}{\sloppy}{}

\title{Constraint-Based Type-Directed Program Synthesis}
\iftoggle{TR}{%
  \subtitle{Extended Report}
  \subtitlenote{This is an extended version of a paper that originally appeared in TyDe 2019~\cite{osera_constraint-based_2019}.}
}{}


\author{Peter-Michael Osera}
\orcid{nnnn-nnnn-nnnn-nnnn}             
\affiliation{%
  \department{Department of Computer Science} 
  \institution{Grinnell College}        
  \country{United States of America}    
}
\email{osera@cs.grinnell.edu}           

\begin{abstract}
  We explore an approach to type-directed program synthesis rooted in constraint-based type inference techniques.  By doing this, we aim to more efficiently synthesize polymorphic code while also tackling advanced typing features such as GADTs that build upon polymorphism.  Along the way, we also present an implementation of these techniques in \scythe, a prototype live, type-directed programming tool for the Haskell programming language and reflect on our initial experience with the tool.
\end{abstract}

\ccsdesc[500]{Software and its engineering~Semantics}
\ccsdesc[500]{Software and its engineering~Automatic programming}
\ccsdesc[300]{Theory of computation~Logic and verification}

\keywords{Functional Programming, Program Synthesis, Type Inference, Type Theory}  

\maketitle

\section{Introduction}%
\label{sec:introduction}
Functional programmers frequently comment how richly-typed functional programs just ``write themselves''.\footnote{Once you get over the complexity of the types!}
For example, consider writing down a function that obeys the type:
\begin{center}
  \begin{minipage}{0.25\textwidth}
    \begin{haskellcode}
f :: a -> Maybe a -> a
    \end{haskellcode}
  \end{minipage}
\end{center}
Because the return type of the function, \hask{a}, is polymorphic we can only produce a value from two sources:
\begin{enumerate}
  \item The first argument to the function (of type \hask{a}).
  \item The result of pattern matching on the second argument (in the \hask{Just} case, its argument will have type $a$).
\end{enumerate}
These restrictions highly constrain the set of valid programs that will typecheck, giving the impression that the function writes itself as long as the programmer can navigate the type system appropriately.
To do this, they must systematically check the context to see what program elements are relevant to their final goal and try to put them together into a complete program.
However, such navigation is usually mechanical and tedious in nature.
It would be preferable if a tool automated some or all of this \emph{type-directed development} process.

Such luxuries are part of the promise of type-directed program synthesis tools.
Program synthesis is the automatic generation of programs from specification.
In this particular case, the specification of our program is its rich type coupled with auxiliary information, \eg, concrete examples of intended behavior.

\subsection{From Typechecking to Program Synthesis}

Type-directed synthesis techniques search the space of possible programs primarily through a reinterpretation of the programming language's type system~\cite{osera_type-and-example-directed_2015, frankle_example-directed_2016, polikarpova_program_2016}.
Traditionally, type systems are specified by defining a \emph{relation} between a context, expression, and type, \eg, $\Gamma \vdash e : \tau$ declares that $e$ has type $\tau$ under context $\Gamma$.
From this specification, we would like to extract a typechecking algorithm for the language.
However, because relations do not distinguish between inputs and outputs, it is sometimes not clear how to extract such an algorithm.
Bi-directional typechecking systems~\cite{pierce_local_2000} alleviate these concerns by making it explicit which components of the system are inputs and outputs of the system, typically by distinguishing the cases where we \emph{check} that a term has a type from the cases where we \emph{infer} that a term has a type.
In the former case, the type acts as an \emph{input} to the system where in the latter case, it is an \emph{output}.

In both cases, the term being typechecked serves as an input to the typechecking algorithm.
However, with type-directed program synthesis, we instead view the term as an output and the type as an input.
For example, consider the standard function application rule found in most type systems:
\[
\inferrule
  {%
    \Gamma \vdash e_1 : \tau_1 \rightarrow \tau_2 \\
    \Gamma \vdash e_2 : \tau_1
  }
  { \Gamma \vdash e_1\;e_2 : \tau_2}
\]
While we can observe that the function application should be typed at the result type of the function $e_1$, it isn't clear which parts of the relations are inputs and outputs and the order in which the checks ought to be carried out.
A bidirectional interpretation of this rule makes the inputs and outputs explicit:
\[
\inferrule
  {%
    \Gamma \vdash e_1 \Rightarrow \tau_1 \rightarrow \tau_2 \\
    \Gamma \vdash e_2 \Leftarrow \tau_1
  }
  { \Gamma \vdash e_1\;e_2 \Rightarrow \tau_1 \rightarrow \tau_2 }
\]
Here, the relation $\Gamma \vdash e_1 \Rightarrow \tau_1 \rightarrow \tau_2$ means that we infer that the type of $e_1$ is $\tau_1 \rightarrow \tau_2$.
The relation $\Gamma \vdash e_2 \Leftarrow \tau_1$ means that we check that the type of $e_2$ is indeed $\tau_1$.
With this, the procedure for type checking a function application is clear: infer a function type $\tau_1 \rightarrow \tau_2$ for $e_1$ and then check that $e_2$ has that input type $\tau_1$; the type of the overall application is then inferred to be the output type $\tau_2$.

With type-directed synthesis, we turn type\-checking into \emph{term generation} by reinterpreting the inputs and outputs of the typechecking relation.
\[
\inferrule
  {%
    \Gamma \vdash \tau_1 \rightarrow \tau_2 \Rightarrow e_1 \\
    \Gamma \vdash \tau_1 \Rightarrow e_2
  }
  { \Gamma \vdash \tau_2 \Rightarrow e_1\;e_2}
\]
The relation $\Gamma \vdash \tau \Rightarrow e$ asserts that whenever we have a \emph{goal type} $\tau$ we can generate a term $e$ of that type.
Now our term generation rule for function application says that whenever we have a goal type $\tau_2$, we can generate a function application $e_1\;e_2$ where the output of the function type agrees with the goal.
We can apply this pattern to the other rules of the type system to obtain a complete term generation system for a language.
We can further augment the system with type-directed example decomposition~\cite{osera_type-and-example-directed_2015, frankle_example-directed_2016} or richer types such as refinements~\cite{polikarpova_program_2016} to obtain true program synthesis systems.

This type-theoretic interpretation of program synthesis gives us immediate insight into how to synthesize programs for languages with rich type systems.
Because rich types constrain the set of possible programs dramatically, program synthesis with types can lead to superior synthesis performance.
On top of this, the type-directed synthesis style directly supports type-directed programming, a hallmark of richly-typed functional programming languages.

\subsection{The Perils of Polymorphism}

However, in supporting rich types, in particular polymorphism, we run into a pair of problems that deserve special attention.
\begin{itemize}[itemsep=0pt]
  \item \textbf{The type enumeration problem}: when dealing with polymorphic types, we must first instantiate them.
  However, there may be many possible instantiations of a polymorphic type.
  For example, consider the \hask{map} function of type \hask{(a -> b) -> [a] -> [b]}.
  Even if we know that we need to create a value of type \hask{[Bool]}, there is nothing directly constraining the type variable \hask{a}.

  Current systems that handle polymorphic synthesis simply enumerate and explore all the possible types that can be created from the context.
  However, this may lead to an excessive exploration of type combinations that do not work, \eg, choosing \hask{a} to be \hask{Int} but not having any functions of type \hask{Int -> Bool} available in the context.
  Furthermore, it may lead to repeated checking of terms that are, themselves, polymorphic, \eg, the empty list \hask{[] :: [a]} for any type \hask{a}.
  We would like to develop a system that systematically searches the space of polymorphic instantiations while minimizing the work done as much as possible.
  \item \textbf{Reasoning about richer types}: polymorphic types are relatively easy to handle in an ad hoc fashion.
  However, polymorphic types form the basis for a variety of advanced type features such as generalized algebraic datatypes (GADTs) and typeclasses that are commonly used in advanced functional programming languages.
  Rather than developing ad hoc solutions for all these related features, it would be useful to have a single framework for tackling them all at once.
\end{itemize}

\subsection{Outline}

In this paper we present a solution to the problems described above: a constraint-based approach to type-directed program synthesis.
Constraint-based typing is used primarily to specify type inference systems which generate constraints between types in the program and then solves those constraints to discover the types of unknown type variables.
In the spirit of type-directed program synthesis, we flip the inputs and outputs of the constraint-based typing system to arrive at a synthesis system that tracks type constraints throughout the synthesis process.
This embodies a new strategy for synthesis---``infer types while synthesizing''---that allows for efficient synthesis of polymorphic code while also giving us a framework to tackle GADTs and other advanced type features built on polymorphism.

In \autoref{sec:constraints}, we develop a constraint-based program synthesis based on a basic constraint-based type inference system.  In \autoref{sec:gadts}, we extend the system to account for GADTs.  In \autoref{sec:scythe}, we discuss our prototype implementation of this approach in our program synthesis tool, \scythe.  And finally in \autoref{sec:extensions} and \autoref{sec:related} we close by discussing future extensions and related work.
\iftoggle{TR}{%
    The extended version of this report which initially appeared in TyDe 2019~\cite{osera_constraint-based_2019} includes several bug fixes to the system as well as complete proofs of soundness in \autoref{app:formalism}.
}{}

\section{Constraint-based Program Synthesis}
\label{sec:constraints}
\begin{figure*}
  $\begin{tabu}{lcll}
  \sigma & \bnfdef & \forall \many{i}{a_i}.\,\tau & \text{Type Schemes} \\
  \tau & \bnfdef & \alpha \bnfalt a \bnfalt T\;\many{i}{\tau_i} \bnfalt \tau_1 \rightarrow \tau_2 & \text{Monotypes} \\
  e & \bnfdef & x \bnfalt K \bnfalt \lambda x.\,e \bnfalt e_1\;e_2 \bnfalt \kwd{case}\;e\;\kwd{of}\;\many{j}{m_j} & \text{Expressions} \\
  m & \bnfdef & K\;\many{k}{x_k} \rightarrow e & \text{Match Branches} \\
  b & \bnfdef & x = e & \text{Top-level Bindings} \\
  \\
  C & \bnfdef & \{ \many{i}{\tau_i \sim \tau_i'} \} & \text{Constraints} \\
  \theta & \bnfdef & \cdot \bnfalt [\tau / a] \theta \bnfalt [\tau / \alpha] \theta & \text{Type Environments} \\
  \Gamma & \bnfdef & \cdot \bnfalt x{:}\sigma, \Gamma \bnfalt K{:}\forall \many{i}{a_i}.\,C \Rightarrow \many{k}{\tau_k} \rightarrow T\;\many{i}{a_i}, \Gamma & \text{Contexts} \\
\end{tabu}$

  \caption{Basic syntax of the object language.}
  \label{fig:basic-syntax}
\end{figure*}

We first develop a constraint-based program synthesis system for a core functional programming language featuring algebraic datatypes and polymorphism.
To do this, we first present a standard Hindley-Milner style type inference system for the language and then show how to re-interpret it as a program synthesis system, augmented with example propagation in the style of \textsc{Myth}~\cite{osera_type-and-example-directed_2015}.

\autoref{fig:basic-syntax} gives the syntax of the object language.
Types are composed of top-level type schemes $\sigma$ which introduce universally quantified type variables.
The language of monotypes $\tau$ enclosed in type schemes include type variables $a$, function types, and saturated algebraic datatypes $T\;\many{j}{\tau_j}$.
The constraint-generation process also generates unification variables $\alpha, \beta, \gamma, \ldots$, that are eventually substituted away during the inference process.
The context $\Gamma$ records both the type of variables and constructors $K$, enforcing that all constructors are (possibly) polymorphic functions whose co-domain is their associated algebraic datatype.
The co-domain is restricted to the universally quantified type variables of the declaration; in \autoref{sec:gadts}, we lift this restriction to enable generalized algebraic datatypes.

The expression language is a standard typed, functional language with algebraic datatype constructors and pattern matching restricted to one-level peeling of head constructors from their arguments.
Like Haskell, we treat constructors like function application and allow for partial application of constructor values.
Conspicuously absent from the expression language are let-bindings.
This is because we do not synthesize let-bindings (that are not top-level) during the synthesis process.
As we shall see shortly, there is no type-or-example-directed way to synthesize such bindings in the general case, so we elide them from our object language.

Throughout this paper, we refer to sequences of objects either using ellipses, \eg, $e_1 \ldots e_k$, or using overbar notation $\many{k}{e_k}$ where we use index variables $i, j, k, \ldots$ to represent the lengths of these sequences.
We denote nested sequences (\ie, sequences of sequences) using nested overbars and both subscripts and superscripts to separate the lengths.
For example, $\many{i}{\many{j}{e^i_j}}$ refers to a sequence (of size $i$) of sequences (of size $j$) of expressions $e$ denoted $e^i_j$.
Finally, whenever possible we use the metavariables $i$, $j$, and $k$ to denote the lengths of the following sequences:
\begin{itemize}[itemsep=0pt]
  \item $i$: types, \eg, $\forall \many{i}{a_i}.\,\tau$, $T\;\many{i}{\tau_i}$,
  \item $j$: examples, \eg, $\many{j}{ \{ S_j \rightarrow \chi_j \} }$, and
  \item $k$: arguments (to constructors), \eg, $K\;\many{k}{e_k}$.
\end{itemize}

\subsection{Type Inference}

\begin{figure}
  \input{figs/inference.tex}
  \caption{Type inference system.}
  \label{fig:inference-system}
\end{figure}

\autoref{fig:inference-system} gives the type inference rules for the system.
As is standard, our constraint-based type inference system operates in two phases: type-and-constraint generation and constraint solving.
The primary judgment of the system $\Gamma \vdash e : \tau \Rightarrow C$ infers that expression $e$ has type $\tau$ under context $\Gamma$ and generates constraints $C$.
Constraints in this system $\tau_1 \sim \tau_2$ assert equalities between types, in particular, giving concrete types to \emph{unification variables} $\alpha$ that are generated for unknown types during the inference process.
For example, in the case of an un-annotated lambda (\textsc{infer-lam}), we do not immediately know the type of the argument to the lambda $x$.
Thus, we create a fresh unification variable $\alpha$ and use that variable as the type of the argument.
In inferring the type of the body of the lambda, we will generate constraints $C$ that refine the type of $\alpha$.

In addition to assigning unification variables to types, constraints also serve the purpose of asserting expected equalities between types of sub-expressions.
When inferring the type of a function application (\textsc{infer-app}), we assert that the type of the head expression is indeed a function type and the argument expression's type is the domain of that function type ($\tau_1 \sim \tau_2 \rightarrow \alpha$).
When inferring the type of a case scrutinee (\textsc{infer-case}), we assert that the type of the scrutinee is indeed an algebraic datatype with fresh unification variables in place of its arguments ($\tau \sim T\;\many{i}{\alpha_i}$).
And finally, we add constraints when inferring the overall type of a case expression stating that all of the branches have the same type ($\many{j}{\beta \sim \tau_j}$).

Polymorphic types interact with the inference system in two ways.
The first is \emph{let-generalization} which occurs at top-level bindings $f = e$.
After generating constraints $C$ for the body of the binding, we solve them using a standard unification algorithm to find a type substitution $\theta$ for each of the unification variables found in $C$ (\textsc{infer-bind}).
Those unification variables with no substituting types ($\forall \many{i}{\alpha_i} \notin \mathrm{fuvs}(\theta)$) become universal type variables in the final type scheme of $f$.
The second way lies in inferring the monotypes we use to instantiate the type schemes of variables.
Rather than requiring that the user provides the instantiation via a type application form or guessing the instantiation upfront, we instantiate a type scheme with fresh unification variables $\many{i}{\alpha_i}$ that will be later constrained based on how the the variable $x$ is used by the program (\textsc{infer-var}).

\subsection{Constraint-based Synthesis}

\begin{figure*}
  $\begin{tabu}{lcll}
  v & \bnfdef & c \bnfalt \lambda x.\,e \bnfalt K\;\many{i}{v_i} & \text{Values} \\
  \chi & \bnfdef & c \bnfalt K\;\many{i}{\chi_i} \bnfalt v \Rightarrow \chi & \text{Examples} \\
  X & \bnfdef & \forall \many{i}{c_i{:}a_i}.\,\many{j}{\chi_j} & \text{Top-level Examples} \\
  S & \bnfdef & \many{i}{[v_i / x_i]} & \text{Environments} \\
  W & \bnfdef & \many{j}{[S_j \mapsto \chi_j]} & \text{Example Worlds} \\
  \Gamma & \bnfdef & \cdots \bnfalt c{:}a, \Gamma & \text{Contexts}
\end{tabu}$

  \caption{Syntactic extensions for type-and-example-directed program synthesis.}
  \label{fig:synthesis-syntax}
\end{figure*}

\autoref{fig:synthesis-syntax} gives the syntactic extensions to the core language to support type-and-example-directed program synthesis in the style of \textsc{Myth}~\cite{osera_type-and-example-directed_2015, frankle_example-directed_2016}.
We extend the language with example values $\chi$ that the user provides as additional specification to the synthesizer.
Intuitively, example values are specifications of values at a particular type that can be decomposed into specifications for those values' components.
For example, each of the values $\ol{v_i}$ of a saturated constructor example $K\;\ol{v_i}$ can become examples for synthesizing each of $K$'s arguments.

However, it is not immediately clear how to decompose values at function type, \ie, lambdas, in a similar way.
Like \textsc{Myth}, we introduce \emph{input-output pairs} $v \Rightarrow \chi$ as a distinguished example value at function type.
Such examples are easily decomposed and align with our intuition that we would like to specify the behavior of a function through a collection of input-output examples (either realized as test cases or documentation).

The synthesis system takes as input a context $\Gamma$, a goal type $\tau$, and a set of examples, and produces a program of that type that agrees with the set of examples.
Ultimately, the system decomposes the examples, assigning sub-values to generated binders---from lambdas and case expressions---as the synthesizer generates them.
These are recorded alongside the examples that generated them, leading to the notion of an \emph{example world} $W$ as a pair of a (value) environment $S$ and a \emph{example goal} $\chi$ for that world.
When checking that a program agrees with these examples, we evaluate the program closed with each example world's environment and check that the resulting value is consistent with that world's example goal.
We write the application of an environment $S$ to an expression $e$ using juxtaposition: $Se$.

\begin{figure}
  \input{figs/refinement-std.tex}
  \input{figs/generation.tex}
  \caption{Type-and-example program synthesis rules.}
  \label{fig:synthesis-system}
\end{figure}

\begin{figure}
  \input{figs/auxiliary.tex}
  \caption{Auxiliary definitions for the synthesis system.}
  \label{fig:synthesis-aux}
\end{figure}

\autoref{fig:synthesis-system} describes the complete synthesis system over the language.
In the spirit of bidirectional typechecking~\cite{pierce_local_2000}, we divide type-directed synthesis into two processes which are realized as a pair of judgments in the system:
\begin{enumerate}[itemsep=0pt]
  \item \emph{Refinement}, $\Gamma \vdash \tau \rhd W \Rightarrow e$, decomposes a set of input examples according to their types.
  \item \emph{Generation}, $\Gamma \vdash \tau \Rightarrow e$ enumerates terms in a type-directed fashion in the absence of examples.
\end{enumerate}
The separation is, in fact, a division of the syntax of the language into \emph{introduction} forms and \emph{elimination} forms for refinement and generation, respectively.
From a type inference perspective, type information flows \emph{into} introduction forms (\ie, we check against their expected types based on their shape) and \emph{out of} elimination forms (\ie, we infer their types based on their sub-components).
From a synthesis perspective, this means we can use the shape of an introduction form to determine how to decompose examples whereas we have no such affordances in the general case for elimination forms.
Thus, we must resort to raw term enumeration for those forms.

As a final note, the system as described does not support recursive functions (note that in \autoref{fig:synthesis-system}, \textsc{refine-lam} does not bind anything for the function itself) in order to make clear the details of constraint propagation in the synthesis process.
Recursion can be easily integrated into the system in the same manner as \textsc{Myth} by treating the set of input-output examples for a function as a \emph{partial function} and using that function value as a binding for the function-being-synthesized.  This requires that the example set is trace complete so that any call to the function has a known value~\cite{osera_type-and-example-directed_2015} and is how the \scythe system discussed in \autoref{sec:scythe} implements recursion.

\subsection{Refinement}
\label{sec:refinement}

The non-polymorphic refinement rules are largely identical to those found in \textsc{Myth}.
The function of each rule is to describe how to decompose a set of example worlds according to the goal type's example form.
For lambdas (\textsc{refine-lam}), we decompose a collection of worlds containing input-output examples by assigning the lambda's binder $x$ the input and synthesizing the body of the lambda with the output as the goal.
For constructors (\textsc{refine-data}), if all the examples share the same head $K$, then we can safely synthesize that constructor and divide up the examples' arguments into examples for each of that constructor's arguments.
For example, suppose that we have a constructor $K : \mathsf{Int} \rightarrow \mathsf{Bool} \rightarrow \mathsf{Int} \rightarrow T$ with example worlds $W = w_1, w_2, w_3$:
\begin{align*}
  w_1 &= S_1 \mapsto K\;0\;\mathsf{True}\;1\\
  w_2 &= S_2 \mapsto K\;2\;\mathsf{False}\;3\\
  w_3 &= S_3 \mapsto K\;4\;\mathsf{True}\;5\\
\end{align*}
By \textsc{refine-data}, we will synthesize $K\;\blacksquare_1\;\blacksquare_2 \;\blacksquare_3$ with three synthesis sub-goals for each of $K$'s arguments.
Each of the sub-goals contains example worlds $W_1$, $W_2$, and $W_3$ respectively:
\begin{align*}
  W_1 &= S_1 \mapsto 0, S_2 \mapsto 2, S_3 \mapsto 4 \\
  W_2 &= S_1 \mapsto \mathsf{True}, S_2 \mapsto \mathsf{False}, S_3 \mapsto \mathsf{True} \\
  W_3 &= S_1 \mapsto 1, S_2 \mapsto 3, S_3 \mapsto 5
\end{align*}
Notably with refinement, the fully specified goal type is required as input to the procedure, so there is no need to reason about type constraints at this point in the system\footnote{%
  Note that even with goal types elided, the examples make inference of the goal type trivial due to canonicity of types.
  However, we may want to consider type-directed synthesis in the absence of any examples which we discuss in more detail in \autoref{sec:extensions}.
}.

Unlike other type-directed forms, case-expressions do not immediately imply a particular type.
Indeed the type of a case-expression is exactly the type shared by its match bodies.
Nevertheless, we can refine examples through case expressions as described by the \textsc{refine-case} rule:
\begin{enumerate}[itemsep=0pt]
  \item Generate a scrutinee expression $e$ at some concrete datatype $T\;\tau_j$.
  \item For each constructor $K$ of $T$, create a match for that constructor using the example worlds $\{ \many{j}{S_j \mapsto \chi_j} \} \subseteq W$ where the scrutinee would evaluate to a value with head constructor $K$ (written using the \emph{filtering operator} $W |^K_e$ defined in \autoref{fig:synthesis-aux}).
\end{enumerate}
When generating a scrutinee, we use the term-generation judgment $C;\Gamma \vdash \tau \Rightarrow e ; C'$ which produces a set of type constraints $C'$ under which $e$ typechecks at type $\tau$.
However, because we assume the concrete datatype of type $T\;\many{i}{\tau_i}$ up-front, we know that the type does not contain any unification variables and thus we can ignore $C'$ for now.
In theory this means that an implementation of the system must explore all possible combinations of datatypes to arguments which is the exact problem we are trying to solve with constraint-based synthesis!
But in practice, we heavily restrict the kinds of expressions that appear as scrutinees already, \eg, to direct arguments of the synthesized function, so such type search is tolerable in practice.

The helper judgment $\Gamma \vdash e; \many{i}{\tau_i}; K; \tau \rhd W \Rightarrow m$ is responsible for synthesizing match branches from the scrutinee and candidate constructor.
Its sole rule (\textsc{refine-match}) extracts and binds the evaluated-to constructor values' arguments in the filtered example worlds and continues synthesis of the match's body.

Finally, we connect refinement and term generation with the \textsc{refine-guess} rule which generates a term $e$ and then checks to see if, for each example world $S \mapsto \chi$, that the closed term $Se$ is equivalent to the goal example $\chi$ according to the standard evaluation rules of the language ($e \equiv_\beta v$ whenever $e \longrightarrow^* v$).
The only non-standard case that can arise is the comparison of a lambda (a $v$) to an input-output example (a $\chi$).
We thus define:
\[
  \lambda x{:}\tau.\,e \equiv_\beta v \Rightarrow \chi \;\mathrm{iff}\; (\lambda x{:}\tau.\,e)\;v \equiv_\beta \chi. \\
\]
That is, we run the input of the input-output pair through the lambda and check to see the resulting value is equivalent to the output of the pair.

\paragraph{Polymorphism in Refinement}
We must consider polymorphic types at two points within refinement.
First, unlike the presentations of the case rule in prior work~\cite{osera_type-and-example-directed_2015}, constructors now have (potentially) polymorphic types.
However, because the type $T\;\many{i}{\tau_i}$ is fully concrete, we can immediately instantiate the constructor's polymorphic type ($\theta = [\many{i}{\tau_i / a_i}]$) and use that type substitution on each of the types of the arguments of the constructor.

Secondly, top-level bindings are at type schemes $\sigma$, so we must provide examples at polymorphic type for them.
However, there is no type-directed syntactic form for polymorphic types!
Furthermore, even an explicit term-level type abstraction, \eg, $\Lambda a.\,\chi$, does not provide much help!
This is because the type abstraction does not introduce \emph{values} of the type variable $a$, and we need such values to write down complete examples.

We therefore introduce a new top-level example value form $X$ in \autoref{fig:synthesis-syntax}, the \emph{polymorphic example} $\forall \many{i}{c_i{:}a_i}.\,\many{j}{\chi_j}$, which introduces a set of \emph{polymorphic constants} $c_i$ at type variables $a_i$.
This example form was first introduced in Osera's thesis~\cite{osera_program_2015} in the context of synthesizing within System F, and we have adapted it to this Haskell-like setting.
For example, we might specify for a goal type of $\forall a.\,a \rightarrow \kwd{Maybe}\;a \rightarrow a$:
\begin{align*}
  \forall c_1{:}a, c_2{:}a.&\,c_1 \Rightarrow \kwd{Nothing} \Rightarrow c_1,  \\
                           &\,c_1 \Rightarrow \kwd{Just}\;c_2 \Rightarrow c_2.
\end{align*}
In these two examples, we introduce polymorphic constants $c_1$ and $c_2$ of type $a$ and use them throughout the input-output examples for the standard $\hask{fromMaybe}$ function.\footnote{%
  The astute reader ought to notice that these examples are virtually the definition of $\hask{fromMaybe}$ already!
  We reflect on this fact in more detail in \autoref{sec:scythe}.
}

We refine polymorphic examples introduced at top-level bindings (\textsc{refine-poly}) by simply stripping the top-level forall and synthesizing at the contained sub-example value.
The binding information for the polymorphic constant is only useful for typechecking the examples themselves; they are not necessary for the synthesis process as the example values are simply carried around as values bound to variables in an example world's environment.

\subsection{Generation}

Like refinement, term generation takes as input a goal type.
However, during the course of generation, we may need to instantiate types schemes when generating variables.
Rather than committing to a choice of a complete instantiation of a type scheme upfront (which requires us to systematically search the space of possible types), we infer the instantiation as we generate the complete term.
If type-and-example-directed program synthesis captures the idea of ``evaluating \emph{during} enumeration''~\cite{osera_type-and-example-directed_2015}, then our approach of integrating type inference into the term generation adds the idea of ``inferring \emph{during} enumeration''.

To do so, we ``invert'' the appropriate rules from our type inference system (\autoref{fig:inference-system}) to create a term generation system that infers types during the generation process.
This judgment, written $C; \Gamma \vdash \tau \Rightarrow e; C'$, takes a collection of constraints $C$, a context $\Gamma$, and type $\tau$ as input as produces a program $e$ and an updated constraint set $C'$ as output.
In effect, we thread a set of type constraints throughout the generation process which makes concrete the idea that a choice of a component in one part of an expression might influence choices in the rest of the expression.
It is worthwhile to note that our rules are not an exact inversion of the type inference system defined in \autoref{fig:inference-system} because the generation process is given a concrete goal type to work with as input.

While refinement handles the introduction forms of the language, generation concerns the elimination forms, of which there are only variables and function application.
With variable generation (\textsc{gen-var}), we speculatively choose a variable to generate and then create a new constraint recording that the goal type must be equivalent to the type scheme of the variable instantiated with fresh unification variables ($\tau \sim [\many{i}{\alpha_i / a_i}] \tau'$).
And with function application generation (\textsc{gen-app}), we first speculatively generate a function under the constraint that its domain matches our eventual argument type and the co-domain matches our goal type.
We then generate function arguments under the constraints accrued from generating the function.

While the initial goal type for generation is a concrete type, we will quickly introduce unknown types, \ie, unification variables, through the process.
We eliminate unification variables through unification of the constraints during the generation process (\textsc{gen-unify}).
If our goal type is some $\tau_1$ (presumably a unification variable) and our current constraint set allows us to conclude that it is equivalent to $\tau_2$ ($C \vDash \tau_1 \sim \tau_2$) then we can synthesize at this type $\tau$ instead.

As we generate constraints, we need to take care that these constraints are consistent.
An inconsistent set of constraints asserts that in-equal types are equal, \eg, $\kwd{Int} \sim \kwd{Bool}$.
If at any point we arrive at a set of inconsistent constraints, we know that the currently generated term is not well-typed and should be rejected.
In our formal system, this amounts to using the $\mathsf{consistent}$ helper relation to check that whenever we output a constraint set from a rule that it is indeed consistent.
$\mathsf{consistent}(C)$ simply asserts that a unifier (\ie, type substitution) exists for $C$ as described in~\autoref{fig:synthesis-aux}.

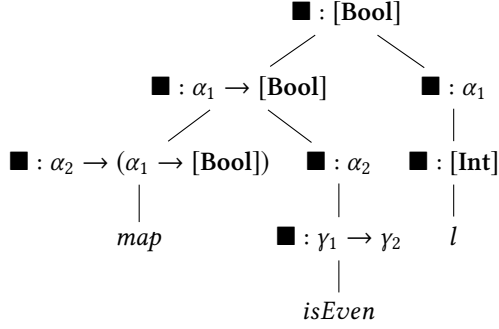
\begin{figure}
  \centering
  \begin{forest}
    [{$\blacksquare : [\kwd{Bool}]$}
      [{$\blacksquare : \alpha_1 \rightarrow [\kwd{Bool}]$}
        [{$\blacksquare : \alpha_2 \rightarrow (\alpha_1 \rightarrow [\kwd{Bool}])$}
          [{$map$}]
        ]
        [{$\blacksquare : \alpha_2$}
          [{$\blacksquare : \gamma_1 \rightarrow \gamma_2$}
            [{$isEven$}]
          ]
        ]
      ]
      [{$\blacksquare : \alpha_1$}
        [{$\blacksquare : [\kwd{Int}]$}
          [{$l$}]
        ]
      ]
    ]
  \end{forest}
  \caption{Example generation derivation of $map$.}
  \label{fig:gen-deriv}
\end{figure}

For example, consider generating a program of type $[\kwd{Bool}]$ under the context:
\begin{align*}
  \Gamma =\;& map : \forall a b.\,(a \rightarrow b) \rightarrow [a] \rightarrow [b], \\
          & isEven : \kwd{Int} \rightarrow \kwd{Bool}, \\
          & l : [\kwd{Int}]
\end{align*}
with goal type $[\kwd{Bool}]$.
\autoref{fig:gen-deriv} graphically shows one potential generation of a program of this type.

Starting at the root, we first apply the \textsc{gen-app} rule which generates a fresh unification variable $\alpha_1$ and two generation sub-goals at types $\alpha_1 \rightarrow [\kwd{Bool}]$ and $\alpha_1$.
We apply \textsc{gen-app} again at the first sub-goal which generates another unification variable $\alpha_2$ and two more sub-goals at types $\alpha_2 \rightarrow (\alpha_1 \rightarrow [\kwd{Bool}])$ and $\alpha_2$.

At the first of these new sub-goals, we can apply \textsc{gen-var} and choose the $map$ variable.
This generates the constraint
\[
  \alpha_2 \rightarrow \alpha_1 \rightarrow [\kwd{Bool}] \sim (\gamma_1 \rightarrow \gamma_2) \rightarrow [\gamma_1] \rightarrow [\gamma_2]
\]
for fresh unification variables $\gamma_1$ and $\gamma_2$.
Now we can return to the sub-goal with goal type $\alpha_2$ and apply \textsc{gen-unify} to continue generating at type $\gamma_1 \rightarrow \gamma_2$ since our constraint implies that $\alpha_2 \sim \gamma_1 \rightarrow \gamma_2$.
This allows us to apply \textsc{gen-var} and choose the $isEven$ variable, adding the constraint
\[
  \gamma_1 \rightarrow \gamma_2 \sim \kwd{Int} \rightarrow \kwd{Bool}.
\]
Finally, we can return to the sub-goal with goal type $\alpha_1$, apply \textsc{gen-unify} to generate at type $[\kwd{Int}]$ (because $\alpha_1 \sim [\gamma_1] \sim [\kwd{Int}]$), and the choose $l$ via \textsc{gen-var}, generating a final, trivial constraint $[\kwd{Int}] \sim [\kwd{Int}]$.

Because we are ensuring that our generated constraint sets are consistent at every step, the order in which we explore generation sub-goals doesn't matter.
However, the order heavily influences how much backtracking we will need to perform because of ``bad'' choices of terms.
For example, if rather than synthesizing $map$ first, we instead try to synthesize its arguments, we may make many spurious choices if there are non-list types in the context.
Thus, even though our constraint-based synthesis system allows us to consider types in a more incremental fashion, we must still be conscious of performance during implementation.
In this particular case, it makes sense to attempt to synthesize functions before their arguments since the function is more likely to prune the search space of types than synthesizing one of its arguments.

\section{Synthesis with GADTs}%
\label{sec:gadts}
Traditional algebraic types require that the output type of its constructors have the form $T\;\many{i}{a_i}$ where the $a_i$s are universally quantified type variables from the type of the constructor.
\emph{Generalized algebraic data types} (GADTs) lift this restriction so that the type arguments of $T$ may be any type, not just type variables. 
While a seemingly insignificant change, GADTs strike a powerful balance between power and complexity on the spectrum of rich type systems.

The quintessential example of GADTs-in-action is the tagless interpreter given below in Haskell:
\begin{minted}{haskell}
data Exp a where
  Lit  :: a -> Exp a
  Plus :: Exp Int -> Exp Int -> Exp Int
  Eq   :: Exp Int -> Exp Int -> Exp Bool
  If   :: Exp Bool -> Exp a -> Exp a -> Exp a

eval :: Exp a -> a
eval (Lit x)       = x
eval (Plus e1 e2)  = eval e1 + eval e2
eval (Eq e1 e2)    = eval e1 == eval e2
eval (If e1 e2 e3) =
  if eval e1 then eval e2 else eval e3
\end{minted}
Here, we parameterize the type of \mintinline{haskell}{Exp a} by the type \mintinline{haskell}{a} that the expression evaluates to.
This allows us to give \mintinline{haskell}{eval} the rich type \mintinline{haskell}{Exp a -> a} which in turn allows us to write \mintinline{haskell}{eval} without any case analyses on the values it produces.

How do GADTs complicate typechecking and consequently type-directed program synthesis?
For example, consider the \mintinline{haskell}{Plus} case of \mintinline{haskell}{eval}.
If \mintinline{haskell}{Plus} was a regular algebraic datatype (say, with type \mintinline{haskell}{Exp -> Exp -> Exp}), the body would not typecheck because \mintinline{haskell}{(+)} produces a value of numeric type whereas the function expects a value of type \mintinline{haskell}{a}.
To enable the body of \mintinline{haskell}{Plus} to typecheck, the typechecker must deduce that because the input is a constructor that produces an \mintinline{haskell}{Exp Int}, then $\hask{a} \sim \hask{Int}$ and thus we must typecheck the body at type \mintinline{haskell}{Int}.
In effect, consuming a GADT via a pattern match introduces \emph{local type constraints} into the environment that we must consider when typechecking each branch~\cite{vytiniotis_outsideinx_2011}.

\begin{figure}
  $\begin{tabu}{lcll}
  \Gamma & \bnfdef & \cdots \bnfalt K{:}\forall \many{i}{a_i}.\,C \Rightarrow \many{k}{\tau_k} \rightarrow T\;\many{i}{a_i}, \Gamma & \text{Contexts}
\end{tabu}$

  \caption{Syntactic extensions for GADT synthesis.}
  \label{fig:gadt-syntax}
\end{figure}

With some minor changes, we can adapt the constraint-based framework we developed previously to accommodate GADTs.
In particular, we follow the presentation of \textsc{OutsideIn}~\cite{vytiniotis_outsideinx_2011} and represent GADTs as standard polymorphic constructors coupled with constraints $C$ on the type variables that appear in the output type of the constructor.
This allows us to avoid a step of unification during type inference and synthesis to discover these constraints on our own.
As an example, we can rewrite the declaration of \mintinline{haskell}{Exp} above in this style:
\begin{minted}{haskell}
data Exp a where
  Lit  :: a -> Exp a
  Plus :: (a ~ Int) =>
            Exp Int -> Exp Int -> Exp a
  Eq   :: (a ~ Bool) =>
            Exp Int -> Exp Int -> Exp a
  If   :: Exp Bool -> Exp a -> Exp a -> Exp a
\end{minted}
With this style, our GADT constructor definitions obey the same restriction as regular algebraic datatypes---the type arguments of the output type are type variables---but the addition of type constraints allow us greater freedom as to what types the constructors produce.

\begin{figure}
  \input{figs/refinement-gadts.tex}
  \caption{Extensions to synthesis system for GADTs.}
  \label{fig:gadt-semantics}
\end{figure}

\autoref{fig:gadt-semantics} gives the updated synthesis rules with generalized algebraic datatypes.
Because GADTs only affect datatype declarations and constructors, we only need to update the rules for refinement; raw term generation is left untouched.

We update the refinement judgment to carry around the set of constraints $C$ gained through case analysis of GADT values.
Unlike term generation, additional constraints gained through GADT analysis are scoped to their matches.
Because of this, we only need to thread constraints down the refinement judgment as an input and not gather updated constraints as an output.
Most rules simply pass their constraints to their sub-refinement calls as in the case of lambdas (\textsc{refine-gadt-lam}) or to raw term generation (\textsc{refine-gadt-guess}).
At the top-level, we begin refinement with the empty set of constraints (\textsc{refine-gadt-binding}).

Now when generating a case expression (\textsc{refine-gadt-case}), we also extract the constraints associated with each constructor $C_i$ and add it to the constraint set associated with synthesizing that constructor's match body.
Notably, we instantiate the additional constraints $C_i$ with the type environment $\theta$ generated by unifying the type of the scrutinee with the output type of the constructors, written $\theta C_i$.
Constraints are utilized during refinement when synthesizing constructors (\textsc{refine-gadt-data}).
We ensure that we only synthesize a constructor if the current set of constraints $C$ implies all of the (instantiated) constraints bundled with the constructor's type.
Note that this stands in contrast to inconsistencies detected during case generation.
If the constraints of a constructor are inconsistent with the current set of constraints when synthesizing a match alternative then the alternative is unreachable and can be elided as an optimization.

Finally, like with term generation, we need to use the constraints in refinement to refine the goal type when appropriate.
The unification rule for refinement (\textsc{refine-gadt-unify}) operates identically to its generation counterpart but at type variables: at a type variable $a$, we can use $a$'s assigned type as recorded by $C$ instead.

As an example of the system in action, consider how constraints are generated and utilized while synthesizing part of the \mintinline{haskell}{eval} function introduced at the beginning of this section.
After an initial application of \textsc{refine-gadt-binding} and \textsc{refine-gadt-lam}, we arrive at the following refinement state:
\begin{align*}
  eval\;e       =\;& \blacksquare :: a \rhd W \\
  \Gamma        =\;& eval : \forall a.\,\kwd{Exp}\;a \rightarrow a, e : \kwd{Exp}\;a \\
  C             =\;& \cdot.
\end{align*}
At this point, we can apply the \textsc{refine-gadt-case} rule to synthesize a case analysis on $e$.
Let's just focus on the $Plus$ constructor which has type:
\[
  \forall b.\,(b \sim \kwd{Int}) \Rightarrow \kwd{Exp}\;\kwd{Int} \rightarrow  \kwd{Exp}\;\kwd{Int} \rightarrow \kwd{Exp}\;b.
\]
(We $\alpha$-rename the constructor's bound type variable to $b$ in order to distinguish it from the type variable from the overall type of the function $a$).
Unifying the return type of $Plus$ with $e$ yields the type substitution $\theta = [a/b]$ which then yields the final constraint $a \sim \kwd{Int}$ which is added to the set of constraints when synthesizing the body of $Plus$'s match at goal type $a$.
This constraint is then passed to term generation which is then used by \textsc{gen-unify} to change the type of the goal from $a$ to $\kwd{Int}$ which then justifies synthesizing an application of the $(+)$ function.
When we go to generate recursive function calls to $eval$, the constraint is also utilized when instantiating the polymorphic type of $eval$ to $\kwd{Exp}\;a \rightarrow a$ with $a \sim \kwd{Int}$.

\subsection{Metatheory}

There are two properties that we ought to consider of the synthesis system we have developed so far:
\begin{itemize}
  \item \textbf{Soundness} states that synthesized programs are consistent with their specification.
  \item \textbf{Completeness} states that we are able to synthesize all programs that meet a particular specification.
\end{itemize}
Here, we briefly discuss these two properties.
\iftoggle{TR}{%
A summary of the complete system as well as detailed proofs can be found in \autoref{app:formalism}.
}{%
A summary of the complete system as well as detailed proofs can be found in the extended version of this paper.
}

\paragraph{Soundness}

The kinds of specification we consider in this system are \emph{types} and \emph{examples}.
Thus, we can consider soundness with respect to each:\footnote{%
  Note that, for brevity's sake, we only state the relevant properties for expression refinement.
  Similar statements must be made for each of the remaining judgments of the system---match synthesis and term generation.
}
\begin{lemma}[Type Soundness]
  If $C; \Gamma \vdash \tau \rhd W \Rightarrow e$ then $C; \Gamma \vdash e : \tau$.
\end{lemma}
\begin{proofsketch}
By induction on the synthesis derivation.
Intuitively, we synthesize terms in a type-directed manner (the synthesis rules are derived directly from the typing rules of the system), so we expect these terms to be well-typed.
\end{proofsketch}

\begin{lemma}[Example Soundness]
  If $C; \Gamma \vdash \tau \rhd W \Rightarrow e$ then for all $S \mapsto \chi \in W$, $Se \equiv_\beta \chi$.
\end{lemma}
\begin{proofsketch}
By induction on the synthesis derivation.
During refinement, we decompose examples in such a way that example satisfaction of sub-components results in example satisfaction of the whole program.
During term generation, the \textsc{refine-gadt-guess} rule ensures example satisfaction directly for generated terms.
\end{proofsketch}

\paragraph{Completeness}
This property states that our synthesis system is capable of generating all possible programs given a fixed goal type and set of examples.
\begin{lemma}[Completeness]
  If $C; \Gamma \vdash e : \tau$ and $\Gamma \vdash W : \tau$ then $C; \Gamma \vdash \tau \rhd W \Rightarrow e$.
\end{lemma}
However, our system does not obey this property in the name of efficiency!
For example, note that \textsc{gen-app} does not allow for the generation of $\mathrm{case}$ expressions in the position of arguments.

But such cases are not terribly interesting with respect to completeness; a case nested in an argument position of an application could instead be ``bubbled up'' to encase the entire application (at the cost of code redundancy).
Are there more interesting programs that the system is incapable of synthesizing?
We leave this exploration to future work.

\section{Integrating Constraint-based Synthesis in Scythe}%
\label{sec:scythe}
We have implemented constraint-based synthesis for polymorphic types as described in \autoref{sec:constraints} in an in-development program synthesis tool called \scythe.
\scythe supports \emph{live, type-directed development} for the Haskell programming language by extending prior work on GHC for typed holes~\cite{gissurarson_suggesting_2018}.
The typed holes facility of GHC allows users to annotate their program with holes.
During typechecking, GHC will report the inferred type of each hole and suggest possible variables from the context to fill these holes.
The search process currently in GHC is fairly na\"{i}ve, only suggesting individual components whose type matches the goal.
\scythe provides a far more in-depth search of the space of possible programs using the program synthesis techniques described in this work.

To do this, \scythe leverages the unmodified GHC API to interact with existing Haskell code with a high degree of compatibility.
It uses GHC to load external libraries and parse, typecheck, interpret, and manipulate Haskell code.
On the surface, this seems difficult to do since the type-and-example-directed program synthesis process interweaves a number of compilation processes, in particular, typechecking and partial evaluation.
However with some engineering techniques and compromises, we can adopt off-the-shelf compiler APIs to provide this rich kind of editing experience, similar to the efforts of Merlin for OCaml~\cite{bour_merlin:_2018}.

As an example of the system in operation, here is the declaration of the \hask{fromMaybe} example from~\autoref{sec:refinement} as realized in the prototype version of \scythe:
\begin{minted}{haskell}
--
-- | {@
--     fromMaybe :: a -> Maybe a -> a
--     fromMaybe a1 Nothing   = a1
--     fromMaybe a1 (Just a2) = a2
--     @@
--     ctx=(Just, Nothing)
--   @}
fromMaybe :: a -> Maybe a -> a
fromMaybe s1 m1 = _
\end{minted}
Scythe looks for special \texttt{\{@...@\}} sections in Haddock comments which denote \scythe example blocks.
These example blocks contain type-and-example information as well as options to control the synthesizer such as the \texttt{ctx} option that allows the user to specify the components to consider during synthesis.
Note that like Haskell, \scythe infers the universal quantifier for polymorphic example values; all variables of the form \hask{ak} where \hask{a} is a type variable introduced by the function signature and \hask{k} is a natural number are considered to be polymorphic constants of type \hask{a}.

After finding the hole, \scythe offers the following single suggestion to complete it:
\begin{minted}{haskell}
> (ok) case m1 of
>       Nothing -> s1
>       Just a1 -> a1
\end{minted}

\subsection{Experience with Polymorphic Synthesis}

To date, we have performed a preliminary investigation of how type-and-example-driven polymorphic synthesis performs within \scythe over small toy functions.
We reflect on our experiences and lessons learned so far from this investigation that are motivating our future work in this area.

\paragraph{Difficulties with Examples}

Polymorphic example values give us great flexibility in how we specify goals of polymorphic type.
However, as discussed in \autoref{sec:refinement}, a collection of polymorphic example values can quickly become the actual function definition itself!
This demonstrates the power of combining polymorphism with example-based reasoning, but it makes the story for a tool like \scythe less compelling at first glance.

This is certainly true when we are forced to give many examples in order to appease the synthesizer.
Early synthesizers such as \textsc{Myth} had the problematic requirement of \emph{trace completeness} that enough examples are provided so that the synthesizer can carry out execution of any recursive function call made during the synthesis process without having the entire function formed.

However, \scythe lifts this restriction by not immediately rejecting candidate programs that fail due to a lack of examples.
A detailed discussion of this feature is beyond the scope of this paper, but the importance of this fact is that \scythe can produce meaningful output with little-to-no examples!
As a result, polymorphic example values become much more appealing to use.
For example, \autoref{fig:scythe-ex-1} shows how the system only requires a single example to synthesize the \mintinline{haskell}{stutter} and \mintinline{haskell}{append} functions over lists.
Note how the examples (1) in contrast to \mintinline{haskell}{fromMaybe}, don't so obviously imply the implementation of the function and (2) resemble the examples that you might naturally write in documentation.

\begin{figure}
\begin{minted}{haskell}
-- | {@
--      stutter :: [a] -> [a]
--      stutter [a1, a2] = [a1, a1, a2, a2]
--      @@
--      recArg=0
--   @}
stutter :: [a] -> [a]
stutter l = _

> (ok) case l1 of
>       [] -> l1
>       (:) a1 l2 -> (:) a1 ((:) a1 (stutter l2))
\end{minted}

\begin{minted}{haskell}
-- | {@
--      append :: [a] -> [a] -> [a]
--      append [a1, a2] [a3, a4, a5] =
--        [a1, a2, a3, a4, a5]
--      @@
--      recArg=0
--   @}
append :: [a] -> [a] -> [a]
append l1 l2 = _

> (ok) case l1 of
>        [] -> l2
>        (:) a1 l3 -> (:) a1 (append l3 l2)
\end{minted}
  \caption{Example \scythe runs for \mintinline{haskell}{stutter} and \mintinline{haskell}{append}}
  \label{fig:scythe-ex-1}
\end{figure}

\paragraph{The Power of Free Theorems}

The guarantees of polymorphic types are well-known~\cite{wadler_theorems_1989} and are one of the primary motivations for our work in automating type-directed programming with synthesis techniques.
Parametricity tells us that the very shape of a polymorphic type determines very precisely the set of programs that inhabit that type.
We witness this in \scythe when we provide \emph{no} examples to the system and rely on raw-term enumeration.\footnote{%
  By default \scythe limits the depth of nested function application and cases to avoid exponential blow-up in the search space.  
}
For example, here is the output from \scythe when no examples are given for \mintinline{haskell}{fromMaybe}:
\begin{minted}{haskell}
> (?) case m1 of
>       Nothing -> s1
>       Just a1 -> a1
> (?) case m1 of
>       Nothing -> s1
>       Just a1 -> s1
> (?) s1
\end{minted}
The question marks in the output denote that while each reported program has the right type, they have not been found to be consistent with any examples (since we have not provided any!).
However, there are so few programs to report due to polymorphism, the user can simply review each implementation to see if it matches the behavior they want for the function.
This stands in contrast to the shallower search that GHC's hole-filling mechanism conducts which only suggests \mintinline{haskell}{s1} as a valid completion.\footnote{%
  As of GHC version 8.6.4.
}

It might be the case that some polymorphic function of interest is too complicated to be specified with examples.
This example shows that synthesis techniques like those used by \scythe still have the potential to provide meaningful feedback, even when only types drive the synthesis process.

\section{Further Extensions}
\label{sec:extensions}
A constraint-based approach to program synthesis opens the doors to a number of additional features that can greatly expand the scope and power of the system.
We briefly reflect on them, suggesting next steps forward based on the techniques developed in this paper.

\paragraph{Existential Types}

One glaring omission from the system described in \autoref{fig:synthesis-system} is the lack of existential types.
In Haskell, existential types are encoded as universal types in constructors that are not mentioned in the output type of the constructor.
The classic example of this are a type used to box elements of a heterogeneous list:
\begin{minted}{haskell}
data Obj where Obj :: Show a => a -> Obj
\end{minted}
From this, we can create our heterogeneous list from a regular list of \mintinline{haskell}{Obj}, \eg, \mintinline{haskell}{[Obj 1, Obj "hi", Obj True]}.
With this set up, we have no way of extracting the values from their \mintinline{Haskell}{Obj} wrappers---the \mintinline{Haskell}{Obj} effectively seals the type variable \mintinline{Haskell}{a} from the outside world---but the typeclass constraint allows us to traverse and \mintinline{Haskell}{show} the values.

When working with GADT types as in \textsc{refine-gadt-case} and \textsc{refine-gadt-data}, we must take care to identify which of the type variables are truly universal (that appear in the output type of the constructor) versus those that are existential (that do not appear in the output type of the constructor).
For example, if we were pattern matching on a value with type \mintinline{Haskell}{Obj}, its sole match would contain a binder for the existentially bound variable \mintinline{Haskell}{a}.
However, we must ensure that we do not unify that variable with a concrete type as its type is really sealed!
This can be accomplished by calculating via a free variable check which type variables are universal or existential and then skolemizing the existential variables so that they cannot participate in the unification process.

On top of this, we must also decide how to deal with examples at existential type.
Can we get away with representing existential types with polymorphic constants and relying on skolemization to assign those constants appropriate types during synthesis or do we need a new example form for existential values?

\paragraph{Typeclass Constraints}

GADTs are only one way that type constraints are introduced in a Haskell program.
More common are \emph{typeclass constraints}, for example, the standard association list \hask{lookup} function of type
\begin{haskellcode}
lookup :: Eq a => a -> [(a, b)] -> Maybe b
\end{haskellcode}
constrains its first argument of type \hask{a} to be an instance of the \hask{eq} typeclass which in turns provides the method \hask{(==)} to compare \hask{a}s with.
Typeclasses are ubiquitous in Haskell code, especially in highly polymorphic code that mixes typeclass constraints with multiparamteter typeclasses and functional dependencies~\cite{jones_functional_1995}.

While seemingly unrelated at first glance, Vytiniotis \etal\;demonstrate how the problems of typeclass constraints and GADTs inference can be solved through their constraint-based type inference system, \textsc{OutsideIn(X)} and implication constraints~\cite{vytiniotis_outsideinx_2011}.
In order to add typeclass constraints to our synthesis system, we will likely need to adapt more of their constraint representation and solving techniques.

\paragraph{Partial Type Information}

One of the goals of this paper is taking the first steps towards understanding how to marry together constraint-based type inference techniques with program synthesis.
However, oddly enough, even though we've adopted a constraint-based approach to synthesis in this work, we still have made the problems of type inference and program synthesis rather separate!
In particular, because we require that the top-level binding of our synthesis problem possesses a fully annotated type, we only need to perform type inference in a very limited setting---instantiating polymorphic types during term generation---which greatly simplified our approach.

A more ambitious marriage of constraint-based type inference and program synthesis would not only leverage constraint-based reasoning techniques in synthesis but truly intertwine type inference and program synthesis.
This would involve removing restrictions of having concrete types everywhere in the system (\eg, the scrutinee of \textsc{refine-case}) and reasoning more about how to incrementally refine goal types in tandem with program discovery.
Such a change would likely impact performance of the synthesizer significantly but give greater flexibility to the user in terms of specifying a program of interest when the types become complicated.

\section{Related Work}
\label{sec:related}
\paragraph{Program Synthesis with Types}

Program synthesis, the automatic generation of programs from specification, is an enormous field encompassing the programming languages, formal verification, software engineering, and most recently the machine learning communities~\cite{gulwani_program_2017}.
Many approaches are used to tackle this problem in a variety of contexts from general-purpose programming to highly-specific domains.  Utilizing types as a guiding principle for synthesis is a small, yet important, piece of the puzzle, especially as it pertains to richly-typed functional programming languages and type-directed programming.
These approaches span the simply-typed, general-purpose domain~\cite{augustsson_[haskell]_2004, osera_type-and-example-directed_2015}, components~\cite{katayama_analytical_2012, feser_synthesizing_2015, gvero_complete_2013, feng_component-based_2017}, and more richly-typed domains~\cite{frankle_example-directed_2016, polikarpova_program_2016}.
Many of these prior efforts handle polymorphism but in an ad hoc manner or using an ``enumerate all possible types'' approach.
The approach that we present in the paper is the first, to our knowledge, to attempt an ``infer while enumerate'' approach to type instantiation.

\paragraph{Type Inference and GADTs}

Our adaption of type inference is heavily rooted in the constraint-based approached proposed by Odersky \etal\;with their $HM(X)$ framework~\cite{odersky_type_1999} as well as the presentation of constraint-based typing found in \emph{TAPL}~\cite{pierce_types_2002}.

GADT type inference, while an undecidable problem~\cite{simonet_constraint-based_2007}, has enjoyed continuous refinement in order to develop inference algorithms that are tractable but also produce useful results~\cite{karachalias_gadts_2015, vytiniotis_outsideinx_2011, chen_principal_2016, jones_functional_1995, moreira_type_2018, simonet_constraint-based_2007, schrijvers_complete_2009}.
While our system does not perform GADT type inference directly, its approach to GADT type checking and constraint passing is inspired by Vytiniotis's \textsc{OutsideIn(X)} framework~\cite{vytiniotis_outsideinx_2011}.

\paragraph{Live Programming}

Our exploration of constraint-based, type-directed program synthesis lives in the context of live programming support for type-directed programming.
This exploration has begun to take shape in the PL community in several forms.
First is the support for language servers which provide efficient compilation services to live programming tools~\cite{bour_merlin:_2018}.
The second is support for hole-based development that was exclusively found in dependently-typed programming languages such as Coq, Agda, and Idris, but are now found in languages like Haskell~\cite{gissurarson_suggesting_2018}.
The third is theoretical explorations of hole-based, interactive development~\cite{omar_hazelnut:_2017, omar_live_2019}.

\begin{acks}
  We thank the TyDe 2019 reviewers for their valuable feedback in improving this manuscript and the students of the Grinnell Pioneer PL research group that have worked on \scythe in its various forms over the last few years: Bogdan Abaev, List Berkowitz, Kevin Connors, Shelby Frazier, Reily Grant, Andrew Mack, Griffin Mareske, Ankit Pandey, Dhruv Phumbhra, Jonathan Sadun, Zachary Segall, and Zachary Susag.

  This material is based upon work supported by the \grantsponsor{nsf}{National Science Foundation}{https://nsf.org} under Grant No.~\grantnum{nsf}{1651817}.
  Any opinions, findings, and conclusions or recommendations expressed in this material are those of the author and do not necessarily reflect the views of the National Science Foundation.
\end{acks}

\bibliography{main}

\iftoggle{TR}{%
  \appendix
  \section{The Formalism}  
  \label{app:formalism}
  \subsection{The System}

We give the complete syntax (\autoref{fig:app-syntax}) and semantics (both typechecking, \autoref{fig:app-typechecking}, and synthesis, \autoref{fig:app-synthesis}) of the constraint-based synthesis system presented in this work.
Auxiliary definitions that support these systems are found in \autoref{fig:app-synthesis-aux}.
This is the final version of the system as presented in \autoref{sec:gadts} which includes generalized algebraic datatypes.

\begin{figure*}[p]
  $\begin{tabu}{lcll}
  \sigma & \bnfdef & \forall \many{i}{a_i}.\,C \Rightarrow \tau & \text{Type Schemes} \\
  \tau & \bnfdef & \alpha \bnfalt a \bnfalt T\;\many{i}{\tau_i} \bnfalt \tau_1 \rightarrow \tau_2 & \text{Monotypes} \\
  \\
  e & \bnfdef & c \bnfalt x \bnfalt K \bnfalt \lambda x.\,e \bnfalt e_1\;e_2 \bnfalt \kwd{case}\;e\;\kwd{of}\;\many{j}{m_j} & \text{Expressions} \\
  m & \bnfdef & K\;\many{k}{x_k} \rightarrow e & \text{Match Branches} \\
  \chi & \bnfdef & c \bnfalt K\;\many{i}{\chi_i} \bnfalt v \Rightarrow \chi & \text{Examples} \\
  X & \bnfdef & \forall \many{i}{c_i{:}a_i}.\,\many{j}{\chi_j} & \text{Top-level Examples} \\
  v & \bnfdef & c \bnfalt \lambda x.\,e \bnfalt K\;\many{i}{v_i} & \text{Values} \\
  b & \bnfdef & x = e & \text{Top-level Bindings} \\
  \\
  C & \bnfdef & \{ \many{i}{\tau_i \sim \tau_i'} \} & \text{Constraints} \\
  S & \bnfdef & \many{i}{[v_i / x_i]} & \text{Environments} \\
  \theta & \bnfdef & \cdot \bnfalt [\tau / a] \theta \bnfalt [\tau / \alpha] \theta & \text{Type Environments/Unifers} \\
  W & \bnfdef & \many{i}{[S_j \mapsto \chi_j]} & \text{Example Worlds} \\
  \Gamma & \bnfdef & \cdot \bnfalt x{:}\sigma, \Gamma \bnfalt K{:}\forall \many{i}{a_i}.\,C \Rightarrow \many{k}{\tau_k} \rightarrow T\;\many{i}{a_i}, \Gamma \bnfalt c{:}a, \Gamma & \text{Contexts}
\end{tabu}$

  \caption{The complete system: syntax.}
  \label{fig:app-syntax}
\end{figure*}

\begin{figure*}[p]
  \input{figs/auxiliary.tex}
  \caption{The complete system: auxiliary synthesis definitions.}
  \label{fig:app-synthesis-aux}
\end{figure*}

\begin{figure*}[p]
  \input{figs/typechecking.tex}
  \caption{The complete system: typechecking.}
  \label{fig:app-typechecking}
\end{figure*}

\begin{figure*}[p]
  \input{figs/refinement-gadts.tex}
  \input{figs/generation.tex}
  \caption{The complete system: synthesis semantics.}
  \label{fig:app-synthesis}
\end{figure*}

\subsection{Metatheory}

As mentioned in \autoref{sec:gadts}, there are two main properties we might want to consider of our synthesis system:
\begin{itemize}
  \item \textbf{Soundness} states that synthesized programs are consistent with their specification.
  \item \textbf{Completeness} states that we are able to synthesize all programs that meet a particular specification.
\end{itemize}

\subsubsection{Soundness}

We prove that the system is sound with respect to the two kinds of specification under consideration: \emph{types} (\autoref{lem:generation-type-soundness} and \autoref{lem:refinement-type-soundness}) and \emph{examples} (\autoref{lem:example-soundness}).

First we introduce a number of auxiliary lemmas about constraints and generation necessary to prove generation type soundness.

\begin{lemma}[Constraint Monotonicity During Generation]
  If $C; \Gamma \vdash \tau \Rightarrow e; C'$ then $C \subseteq C'$.
\end{lemma}
\begin{proof}
  By a straightforward induction on the generation derivation for $e$.
  Note that all of the generation rules only add to the input constraint set and never remove any constraints.
\end{proof}

\begin{lemma}[Constraint Weakening During Solving]
  If $C \vDash \tau \sim \tau'$, $C \subseteq C'$, $\mathsf{consistent}(C)$, and $\mathsf{consistent}(C')$, then $C' \vDash \tau \sim \tau'$.
\end{lemma}
\begin{proof}
  Suppose that $C'$ did not entail the desired equality.
  Because $C \subseteq C'$, then $C'$ must have some additional constraint that $C$ does not have that refutes the equality.
  However, this additional constraint would render $C'$ inconsistent, a contradiction.
\end{proof}

\begin{lemma}[Constraint Weakening During Typing]
  If $C; \Gamma \vdash e : \tau$, $C \subseteq C'$, $\mathsf{consistent}(C)$, and $\mathsf{consistent}(C')$, then $C'; \Gamma \vdash e : \tau$.
\end{lemma}
\begin{proof}
  By a straightforward induction on the generation derivation for $e$.
  In particular, typing depends entirely on constraints only for solving ($C \vDash C'$), and the constraint weakening during solving lemma tells us that this is safe to do so under a weakened (extended) constraint set.
\end{proof}

With these lemmas in hand, we can prove generation type soundness.

\begin{theorem}[Generation Type Soundness]
  \label{lem:generation-type-soundness}
  If $C; \Gamma \vdash \tau \Rightarrow e;C'$ then $C'; \Gamma \vdash e : \tau$.
\end{theorem}
\begin{proof}
  By induction on the generation derivation of $e$.
  Term generation is derived from typing, so we expect the assumptions gained by analysis of the term generation rules to coincide with the proof obligations from the typing rules.
  What remains is ensuring that we manage our constraint sets so that they contain the necessary constraints to type our generated terms appropriately.
  \begin{proofcases}
    \proofcase{gen-var} By assumption, we know that $\tau$ generates $x$, $x{:}\forall \many{i}{a_i}.\,C' \Rightarrow \tau' \in \Gamma$, and $C'' = C \cup \theta C' \cup \{ \tau \sim \theta \tau' \}$ is consistent.
    We must show that $C''; \Gamma \vdash x : \tau$.
    By \textsc{type-unify}, it suffices to show that $C'' \vDash \{ \tau \sim \many{i}{[\alpha_i / a_i]} \tau' \}$ and $C''; \Gamma \vdash x : \many{i}{[\alpha_i / a_i]} \tau'$.
    We know that the former holds by assumption.
    The latter holds by \textsc{type-var} as long as $C'' \vDash \many{i}{[\alpha_i / a_i]}$ which is true since $C''$ is consistent and $\many{i}{[\alpha_i / a_i]}$ is one of its constraints.
    \proofcase{gen-app} By assumption, we know that we generate $e_1$ at type $\alpha \rightarrow \tau$ and $e_2$ at type $\alpha$.
      By our induction hypotheses, we know that $e_1$ has type $\alpha_1$ under constraints $C_1$ and $e_2$ has type $\alpha_2$ under constraints $C_2$.
      By constraint monotonicity, we know that $C \subseteq C_1 \subseteq C_2$.
      And by constraint weakening, since we know that $C_1$ and $C_2$ are consistent, $e_1$ also has type $\alpha_1$ under constraints $C_2$.
      Finally, we can conclude by \textsc{type-app} that $C_2; \Gamma \vdash e_1\;e_2 : \tau$.
    \proofcase{gen-unify} By assumption, we know that $\tau_2$ generates $e$ under constraints $C'$, $C \vDash \tau_1 \sim \tau_2$, and $\mathsf{consistent}(C')$.
      By \textsc{type-unify}, it suffices to show that $C'; \Gamma \vdash e; \tau_2$ and $C' \vDash \tau_1 \sim \tau_2$.
      The former is true directly from our induction hypothesis.
      The latter is true because constraint monotonicity tells us that $C \subseteq C'$ and constraint weakening says that since $C \vDash \tau_1 \sim \tau_2$, $C' \vDash \tau_1 \sim \tau_2$ as well.
  \end{proofcases}
\end{proof}

We can then tackle type soundness of the refinement portions of the system.

\begin{theorem}[Refinement Type Soundness]
  \label{lem:refinement-type-soundness}
  \begin{enumerate}
    \item If $C;\Gamma \vdash \tau \rhd W \Rightarrow e$ then $C;\Gamma \vdash e : \tau$.
    \item If $K{:}\forall \many{i}{a_i}.\,C' \Rightarrow \many{k}{\tau_k} \rightarrow T; \many{i}{a_i} \in \Gamma$ and $C; \Gamma \vdash e; \many{i}{\tau_i}; K; \tau \rhd W \Rightarrow m$ then $C; \Gamma \vdash \many{i}{\tau_i}; m : \tau$.
  \end{enumerate}
\end{theorem}
\begin{proof}
  By mutual induction on the two synthesis derivations.
  Like generation, refinement is derived from typing, so we expect the assumptions gained by the refinement rules to coincide with the proof obligations from the typing rules.
  \begin{proofcases}
    \proofcase{refine-gadt-lam} By assumption, we refine $e$ at type $\tau_2$ and by our induction hypothesis, $e$ is well-typed at $\tau_2$.
      Thus by \textsc{type-lam}, $\lambda x.\,e$ is well-typed at $\tau_1 \rightarrow \tau_2$.
      \proofcase{refine-gadt-data} By assumption, we refine each of the $k$ arguments of $K$, $\many{k}{e_k}$, at their respective types $\many{k}{\tau_k}$.
      By our induction hypotheses, each $\many{k}{e_k}$ is well-typed at $\many{k}{\tau_k}$.
      Thus, by \textsc{type-data}, $K\,\many{k}{e_k}$ is well-typed at $T\;\many{i}{\tau_i}$.
    \proofcase{refine-gadt-case}
      By assumption, we know that $T\;\many{i}{\tau_i}$ generates $e$ under constraints $C'$ and for each of the $j$ constructors $K_j$ of $T$, we refine an appropriate match $m_j$ at type $\tau$.
      By our induction hypotheses, we know that the body of each $m_j$ has type $\tau$, and by \textsc{type-case} we can conclude that the entire $\kwd{case}$ expression has type $\tau$.
    \proofcase{refine-gadt-guess}
      By assumption, we know that $C; \Gamma \vdash \tau \Rightarrow e; C'$ and by refinement type soundness for generation, we know that $e$ has type $\tau$ under assumptions $C'$.
      By constraint monotonicity, $C \subseteq C'$.
      By inspecting our generation rules, we note that the only additional constraints in $C'$ are added in \textsc{gen-var}.
      There are two kinds of constraints added:
      \begin{enumerate}
        \item $\theta C$, the constraint that our overall constraint set satisfies the constraints introduced by the variable.
          $e$ must satisfy this constraint---otherwise we would have encountered an inconsistent constraint set during generation and rejected $e$---and so we can safely remove this constraint from $C'$.
        \item $\tau \sim \theta \tau'$, the constraint that our goal type is equivalent to some instantiation of the variable's polymorphic type.
          $e$'s type by definition does not contain any unification constraints---these are only created during term generation---and so these constraints do not can also be safely removed.
      \end{enumerate}
      Since both kinds of constraints can be safely removed from $C'$, we can conclude that $e$ is also well-typed under $C$.
    \proofcase{refine-gadt-unify}
      By assumption, we refine $e$ at type $\tau$ and $C \cup \{ \tau \sim \alpha \}$ so by \textsc{type-unify}, we can conclude that $e$ is well-typed at type $a$.
    \proofcase{refine-gadt-match}
    By assumption, we know that $\tau$ generates the match $K\;\many{k}{x_k} \rightarrow e'$ under constraints $C \cup \theta C'$ and context $\many{k}{x_k{:}\theta \tau_k}, \Gamma$.
    By our inductive hypothesis, the branch body $e'$ has type $\tau$ and by \textsc{type-match}, the overall match has type $\tau$ as well.
  \end{proofcases}
\end{proof}

And finally we can use \autoref{lem:generation-type-soundness} and \autoref{lem:refinement-type-soundness} to prove the soundness of our top-level synthesis judgment for bindings.

\begin{theorem}[Top-level Binding Type Soundness]
  If $\Gamma \vdash \sigma \rhd X \Rightarrow b$ then $\Gamma \vdash b : \sigma$.
\end{theorem}
\begin{proof}
  By induction on the refinement derivation of $b$.
  There is only one case, \textsc{refine-gadt-binding}.
  By assumption, we know that we refine $\sigma \rhd W$ to $x = e$ under the context $\Gamma$ extended with bindings for the polymorphic constants introduced in $W$ and $\mathrm{ftvs}(\tau) \subseteq \many{i}{a_i}$.
  By refinement type soundness, we know that $e$ is well-typed at some type $\tau$ and thus by \textsc{type-binding}, we know that the overall binding has the desired polymorphic type.
\end{proof}

Examples are only used during refinement, so we need only consider soundness of the refinement portion of the system.

\begin{theorem}[Example Soundness]
  \label{lem:example-soundness}
  If $\Gamma \vdash W : \tau$ and $\cdot;\Gamma \vdash \tau \rhd W \Rightarrow e$ then for all $S \mapsto \chi \in W$, $Se \equiv_\beta \chi$.
\end{theorem}
\begin{proof}
  By induction on the refinement derivation of $e$.
  At a high level, we know refinement is sound because example refinement coincides with a single step of evaluation based on the shape of the expression being refined.
  \begin{proofcases}
    \proofcase{refine-gadt-lam}
      By assumption we know $\lambda x.\,e$ refines at $\tau_1 \rightarrow \tau_2 \rhd \many{j}{\{S_j \mapsto v_j \Rightarrow \chi_j\}}$, $\many{j}{\{S_j \mapsto v_j \Rightarrow \chi_j\}}$ is well-typed at $e$, and $e$ refines at $\tau_2 \rhd \many{j}{\{[v_j / x] S_j \mapsto \chi_j\}}$.
      By \textsc{type-ex-io}, we know that $v$ has type $\tau_1$ and $\chi$ has type $\tau_2$ and thus, we can apply our induction hypothesis to conclude that $e$ satisfies $\many{j}{\{[v_j / x] S_j \mapsto \chi_j\}}$.
      We must show that $\forall j.\,S_j (\lambda x.\,e) \equiv_\beta v_j \Rightarrow \chi_j$; consider a single such example world $S \mapsto v \Rightarrow \chi$.
      By definition $S (\lambda x.\,e) \equiv_\beta v_j \Rightarrow \chi_j$ if $S (\lambda x.\,e) v_j \rightarrow ([v_j /x] S) e\equiv_\beta \chi_j$.
      But we know this because $e$ satisfies $[v / x] S \mapsto \chi$ from our induction hypothesis.
    \proofcase{refine-gadt-data}
      By assumption we know $K\;\many{k}{e_k}$ refines at $T\;\many{i}{\tau_i} \rhd \many{j}{\{ S_j \mapsto K\;\many{k}{\chi^j_k} \}}$, each $e_k$ refines at $\theta \tau_k \rhd \{ S_j \mapsto \many{j}{\chi^j_k} \}$, and the example worlds are well-typed at type $T\;\many{i}{\tau_i}$.
      By \textsc{type-var} and an application of \textsc{type-app} for each of the $k$ arguments of $K$, we know that each $\chi^j_k$ has type $\theta \tau_k$.
      We must show that $\forall j.\,S_j (K\;\many{k}{e_k}) \equiv_\beta K\;\many{k}{\chi_k}$; consider one such example world $S \mapsto K\;\many{k}{\chi_k}$.
      We know that $S_j (K\;\many{k}{e_k}) = K\;\many{k}{S_j e_k}$ so it suffices to show that $\many{k}{e_k} \equiv_\beta \chi_k$.
      However, we know this is true by applying our inductive hypothesis to each $e_k$.
    \proofcase{refine-gadt-case}
    By assumption we know we refine $\kwd{case}\;e\;\kwd{of}\many{j}{m_j}$ at $\tau \rhd W$, we generate $e$ at type $T\;\many{i}{\tau_i}$, and that we refine each of the case matches at $\tau \rhd W$.
      Consider one of these matches $K\;\many{k}{x_k} \rightarrow e'$.
      By \textsc{refine-gadt-match}, we refine the body of this match $e'$ at type $\tau \rhd \many{j}{\{\many{k}{[v^j_k / x_k]} S_j \mapsto \chi_j \}}$ where $W|^K_e = \many{j}{\{S_j \mapsto \chi_j\}}$.
      And by \textsc{type-case} and \textsc{type-match}, we know that $e'$ has type $\tau$.

      Overall, we must show that $\forall j.\,S_j (\kwd{case}\;e\;\kwd{of}\many{j}{m_j}) \equiv_\beta \chi_j$.
      For our chosen match branch with constructor $K$, consider only example worlds $W|^K_e$ where the case expression evaluates to this branch.
      Consider one such example world $S \mapsto \chi \in W|^K_e$.
      By assumption $S e \longrightarrow^* K\;\many{k}{v_k}$.
      And so it is sufficient to show that $(\many{k}{v_k / x_k} S) e' \equiv_\beta \chi$ but we know this by our induction hypothesis.

      Finally note that by construction, we include one match for each possible constructor of datatype $T$.
      Thus, all example worlds are sent to exactly one of the branches of the $\kwd{case}$.
      Therefore, we can conclude that our $\kwd{case}$ expression satisfies the example worlds.
    \proofcase{refine-gadt-guess}
      By assumption, we already know that $S_je \equiv_\beta \chi_j$ for all $j$ example worlds of $W$.
    \proofcase{refine-gadt-unify}
      By assumption, we know that both $a$ and $\tau$ refine to $e$ and by our induction hypothesis, we know that $e$ satisfies $W$.
  \end{proofcases}
\end{proof}

Example soundness of the top-level judgment is an immediate consequence of example soundness of the refinement judgment.

\begin{theorem}[Top-level Example Soundness]
  If $\Gamma \vdash X : \sigma$ and $\Gamma \vdash \sigma \rhd \forall \many{k}{c_k{:}a_k}.\,\many{j}{\chi_j} \Rightarrow x = e$ then $\forall j.\,e \equiv_\beta \chi_j$.
\end{theorem}
\begin{proof}
  By induction on the synthesis derivation for the binding $x = e$.
  The sole rule is \textsc{refine-gadt-binding} which means that we assume that
  $x = e$ is refined by $\forall \many{i}{a_i}.\,\tau \rhd \forall \many{k}{c_k{:}a_k}.\,\many{j}{\chi_j}$, the body $e$ is refined by $\tau \rhd \{ \many{j}{\cdot \mapsto \chi_j} \}$, and the top-level examples have type $\forall \many{i}{a_i}.\,\tau$.
  By \textsc{type-ex-poly}, we know that each $\chi_j$ has type $\tau$ and so our induction hypothesis tells us that $e$ satisfies each $\chi_j$ immediately.
\end{proof}

\subsubsection{Completeness}

\begin{theorem}[Refinement Completeness]
  If $\cdot;\Gamma \vdash e : \tau$ and $\Gamma \vdash W : \tau$ then $\cdot;\Gamma \vdash \tau \rhd W \Rightarrow e$.
\end{theorem}

Completeness does not hold for the system because we cannot synthesize all programs.
Like its predecessors~\cite{osera_type-and-example-directed_2015, frankle_example-directed_2016}, the system searches for programs already in $\beta$-normal, $\eta$-long form, avoiding consideration of obviously redundant programs such as $(\lambda x.\,x)\;\mathsf{True} \longrightarrow^* \mathsf{True}$.
However, as discussed in \autoref{sec:gadts}, the system also does not allow for synthesis of expressions that involve \kwd{case}s or lambdas in argument positions of function applications.
In the case of \kwd{case}-expressions, nested \kwd{case}s can be bubbled up to enclose the expressions that they were originally nested.
However, with lambdas, we might not be able to synthesize some meaningful programs, \eg, an application of \hask{map} to an increment function (in a system extended with support for integers): \hask{map (\x -> x + 1) l}.

This reflects the implementation of the system in \scythe which also does not allow for generation of lambdas in argument positions.
Such generation is far less constrained than variable and function application generation and thus significantly degrades performance.
However, the restriction is merely practical.
One could add in term generation for lambdas and \kwd{case}-expressions which would resemble the refinement judgment defined in \autoref{fig:app-synthesis} but simply omits the portions of judgments and premises related to examples.
The existing proofs for soundness should be easily extended with these new rules and then completeness could be analyzed similarly to the completeness of $\lambda_{\mathsf{syn}}$, a synthesis system for a core lambda calculus~\cite{osera_program_2015}.

}{%
}

\end{document}